\documentclass{llncs}
\usepackage{amsfonts,amssymb,amsmath}

\title{Regular realizability problems and  regular languages}
\author{A. Rubtsov\thanks{Supported in part by RFBR grant
    14--01--00641.}
}

\institute{Moscow Institute of Physics and Technology\\
\and National Research University Higher School of Economics\\
\email{rubtsov99@gmail.com}
}

\date{}
\spnewtheorem{Def}{Definition}{\bfseries}{\upshape}
\spnewtheorem{prop}{Proposition}{\bfseries}{\itshape}
\spnewtheorem{Claim}{Claim}{\bfseries}{\itshape}
\spnewtheorem*{known}{Theorem}{\bfseries\upshape}{\itshape}
\spnewtheorem*{Lemma}{Lemma}{\bfseries\upshape}{\itshape}
\let\eps\varepsilon

\let\leq\leqslant

\let\es\varnothing

\def\L{{\mathbf{L}}}

\def\NL{\ensuremath{\mathbf{NL}}}

\def\A{\ensuremath{\mathcal A}}

\def\B{\ensuremath{\mathcal B}}

\def\reg{\mathrm{RR}}

\def\leGen#1#2{\mathbin{\leq^{\mathrm{#2}}_{\mathrm{#1}}}}

\def\lelog{\leGen{log}{}}

\def\ledfst{\leGen{dfst}{}}

\begin{document}

\maketitle

\begin{abstract} 
We investigate regular realizability (RR) problems, which are the problems
of verifying whether intersection of a regular language -- the input of the
problem -- and fixed language called filter is non-empty. We consider two kind of problems depending on representation of regular language. If a regular language on input is represented by a DFA, then we obtain (deterministic) regular realizability problem and we show that in this case the complexity of regular
realizability problem for an arbitrary regular filter is either $\L$-complete or $\NL$-complete. We also show that in case of representation regular language on input by NFA the problem is always $\NL$-complete.
\end{abstract}

\section{Introduction}

The regular realizability problems are the problems
of verifying whether intersection of a regular language -- the input of the
problem -- and fixed language called filter is non-empty. Filter $F$ is a
parameter of the problem. Depending on representation of a regular
language we distinguish the deterministic RR problems $\reg(F)$ and
the nondeterministic ones $\reg^n(F)$, which are
corresponds to the description of the regular language either by a
deterministic or by a  nondeterministic finite automaton.

The main question of studying regular realizability problems is the investigation of it's algorithmic complexity depending on a filter. Algorithmic complexity of corresponding regular realizability problem is a kind of a complexity measure on languages. Investigation of the problems algorithmic complexity in case of regular filters is a natural question. 
Moreover, the relation between algorithmic complexities of $\reg(F)$ and
$\reg^n(F)$ is still unknown, but only in the case of regular filters we know the separation modulo $\L \neq \NL$ conjecture. Our main result is the separation of regular languages into two classes: the first class contains languages with correspondent deterministic RR-problems belong to class $\L$ and the second class contains languages with correspondent  $\NL$-complete deterministic RR-problems.

Deterministic regular realizability problems corresponds to a computational model, called Generalized Nondeterministic Automata (GNA).  We investigate this model in section \ref{secGNA}.

\section{RR-problems and deterministic finite state transductions} 

In this section we define RR-problems formally and show how its algorithmic complexity relates to deterministic finite state transductions on filters. We study the deterministic version ($\reg(F)$) as the main one, because as we show further the nondeterministic version is too powerful for regular filters -- for all nonempty regular languages corresponding $\reg^n$-problems are $\NL$-complete.

\begin{Def}
	The regular realizability problem $\reg(F)$ is the problem
	of verifying non-emptiness of the intersection of the filter $F$
with a regular language $L(\A)$, where $\A$ is a DFA. 
Formally
	$$\reg(F) = \{ \A \mid \A \text{ is DFA and } L(\A)\cap F
	\neq \es  \}. $$ 
\end{Def}

In the same way we define the nondeterministic version:

$$\reg^n(F) = \{ \A \mid \A \text{ is NFA and } L(\A)\cap F
\neq \es  \}. $$

Since we are going to consider classes $\L$ and $\NL$, we choose the logspace reduction for RR-problems. We say that filter $F_1$ dominates filter $F_2$ if $\reg(F_2) \lelog \reg(F_1)$. The natural goal is to describe dominance relation on filters by some structural properties of languages. Even finding relation on filters, that respects dominance relation is a hard problem: we know only one such relation -- deterministic finite state transduction.  We define this relation as the relation provided by deterministic finite state transducer. First we recall the definition of finite state transducer, which is also known as rational transducer.

Formally a  finite state transducer is defined by tuple 
$T = (A, B, Q, q_0, \delta, F)$, where 
$A$ is the input alphabet, $B$ is the output alphabet, 
$Q$ is the (finite) state set, $q_0$ is the initial
state,  $F \subseteq Q$ is the set of accepting states and 
$\delta\colon Q \times (A\cup\eps) \times
(B\cup\eps)\times Q $ is the
transition
relation. 

As in case of automata, we say finite state transducer to be deterministic if transition relation $\delta$ is a function.

Consider two DFSTs $T_1$ and $T_2$. We say that a DFST
$T= T_1 \circ T_2$ is the composition of $T_1$ and $T_2$ if $T(x) = y$ iff $T_1(x) = u$ and $T_2(u) = y$.

Define the composition of transducer $T$ and automaton $\A$ in the
same way: we say that automaton $\B = T \circ \A$ recognizes language $\{x \,|\,  T(x) = y \in L(\A) \}$. 


The following proposition is an algorithmic version of Elgot-Mezei
theorem (see, e.g., \cite[Th. 4.4]{Be09}).

\begin{prop}\label{compose}
 The composition of transducers and the composition of a transducer
   and an automaton are computable in deterministic log space.
\end{prop}

We say that filter $F_1$ \emph{covers} filter $F_2$ if there exists such dfst $T$, that $F_2 = T(F_1)$. We also write this as $F_2 \ledfst F_1$.

\begin{lemma}\label{transductions-vs-logspace-reductions}
	If $F_1 \ledfst F_2$ then $\reg(F_1) \lelog \reg(F_2) $.
\end{lemma}
\begin{proof}
	Let $T$ be a deterministic finite state transducer such that $F_1 = T(F_2)$
	and let $\A$ be an input of the $\reg(F_1)$ problem. Build the
	automaton $\B = T \circ \A $ and use it as an input of the
	$\reg(F_2)$ problem. 
	It gives the log space reduction due to
	Proposition~\ref{compose}.  
\end{proof}

We obtain similar results on relation between $\reg^n$ problems and rational (not necessary deterministic) transductions in the paper about RR-problems and context-free languages \cite{RV_DCFS2015}.

Now we divide the class of regular languages into two parts. We call regular filter $F$ \emph{hard} if $F$ covers an arbitrary regular language. It means that for every regular language $R$ there exists DFST $T$, such that $R = T(F)$; in the other case we call regular filter $F$ \emph{easy}.

\begin{prop}\label{REGtypes}
	Regular language $F$ is hard iff an arbitrary deterministic automaton, recognizing $F$, has paths $q \xrightarrow{u} q $, $q \xrightarrow{v} q $, for some state $q$ and some words $u \neq vw, v \neq uw$.
\end{prop}
\begin{proof}
	Consider automaton $\A$, recognizing $F$, such that from each state $s$ there is a path to some accepting state. By the condition of the proposition there is some state $q$ and words $u, v$, such that $q \xrightarrow{u} q $, $q \xrightarrow{v} q $, $u \neq v$, moreover $u$ is not a prefix of $v$ and $v$ is not a prefix of $u$. Let $p$ be the path $q_0 \xrightarrow{p} q$ from the initial state to state $q$ and $s$ be the pass $q \xrightarrow{s} q_f$ from state $q$ to some accepting state $q_f$.

	First we build DFST $T$, which maps $F$ to $\Sigma^*$. Assume, without loss of generality, that we work with the binary alphabet $\Sigma = \{a,b\}$. Describe the construction of $T$. First transducer $T$ expects the word $p$ on the input, processes it and writes nothing. Than if $T$ reads word $u$, it writes letter $a$, if $T$ reads word $v$, it writes letter $b$. Since $u$ is not a prefix of $v$ and $v$ is not a prefix of $u$, transducer $T$ writes $a$ and $b$ independently. If $T$ reads word $s$ it goes to the (only) accepting state.

It is clear that $T(F) = \Sigma^*$. Now it is easy to build DFST $T_R$, which maps $\Sigma^*$ to regular language $R$ and to build the composition $T \circ T_R$ as a resulting transducer $T^\prime$, which maps $F$ on $R$. To build $T_R$ we take an arbitrary DFA $\A$, recognizing $R$ and turn it to DFST $T_R$ by adding output tape and modifying transition function by adding to output the letter of transition: if the automaton has transition $\delta_\A(q,a) = q^\prime$, then the transducer has transition  $(q,a,a,q^\prime) \in \delta_{T_R}$.

Let us call cycle the path of form $q \xrightarrow{u} q$. Now we prove that if there exists an automaton, recognizing $F$, without distinct cycles, than there is no transducer $T$, that maps $F$ to $\Sigma^*$.  Notice that in this case language $F$ can be described by regular expression consisting of finite union of expressions of form $px_1^*y_1x_2^*y_2\ldots x_n^*y_ns$ -- since there are no two distinct cycles, if state  $q_i$ has some nonempty cycle, then there is word $x_i$, such that each path $q_i \xrightarrow{w} q_i$ can be described as $w = x_i^k$.  It is easy to see, that in case of one expression of form $px^*s$, language $F$ doesn't cover $\Sigma^*$. Indeed, if there is a transducer $T$, such that $T(px^*s) = \Sigma^*$ then for long enough word $w$ from $\Sigma^*$ should exist a long enough word from $px^*s$. But since transducer $T$ has finitely many states there are such numbers $n$ and $k$, that $q_0 \xrightarrow{px^n} q$ and $q_0 \xrightarrow{px^{n+k}} q$, so we get that each long enough word from $T(px^*s)$ has a periodic subword, and come to contradiction -- there are words without periodic factors. Using the pigeon-hole principle again we obtain, the same contradiction in general case: for the expressions of form $px_1^*y_1x_2^*y_2\ldots x_n^*y_ns$ and thus we obtain the same contradiction for their finite union. Finite union is contained in some regular language of form $w_1^*w_2^*\ldots w_n^*$ and since that language does not cover $\Sigma^*$, finite union also doesn't.

\end{proof}

Recall, that regular language $R$ is called \emph{bounded} if there are such words $w_1,\ldots w_n$,  that $R \subseteq w_1^*w_2^*\ldots w_n^*$. From the proof of proposition \ref{REGtypes} we obtain the corollary.

\begin{corollary}
	Regular language is easy iff it is a bounded regular language.
\end{corollary}

\begin{remark}\label{boundeddfst}
	If regular language $F$ is easy, then $F \ledfst w_1^*w_2^*\ldots w_n^*$.
\end{remark}
\begin{proof}
	Indeed, since $F \subseteq w_1^*w_2^*\ldots w_n^*$, we just apply DFST $ID_F$ (which maps $w$ to $w$ iff $w\in F$) to the language $w_1^*w_2^*\ldots w_n^*$.
\end{proof}

\section{Generalized Nondeterministic Automata}\label{secGNA}

Regular realizability problems have corresponding computational model, called Generalized Nondeterministic Automata. By \emph{generalized nondeterministic automaton} $M_F$   (depending on filter $F$) we mean determinstic logspace Turing machine with advanced one-way read-only tape. GNA $M_F$ accepts word $w$ if there exists such word $adv \in F$, that $M$ accepts $w$ when $adv$ is writen on the advanced tape. We call the advanced tape \emph{advice tape}.

Let us describe the model a little more formally. Advice tape has an alphabet $\Delta$, and blank-symbol $\Lambda \not\in \Delta$. We define $M_F(w, \alpha)$ to be the function of two arguments -- word $w$ on the input tape and word $\alpha \in F \subseteq \Delta^*$ on the advice tape. After word $\alpha$ advice tape is filled with blank-symbols $\Lambda$. On each step GNA can move the head of advice tape and read the next symbol or not to move the head. We assume, that filter is always a non-empty language. So, the function $M_F(w, \alpha)$ equals $1$ if GNA $M_F$ reaches some accepting configuration on pair $(w,\alpha)$ and in the other case we assume that $M_F(w,\alpha) = 0$ (we assume, that GNA stops on each pair). We say that GNA $M_F$ accepts word $w$ if there is such advice $\alpha \in F$, that $M_F(w,\alpha) = 1$ and as usual by language $L(M_F)$ we mean all the words, accepted by $M_F$.

Why do we use the word automata? First, GNA was defined in \cite{VyaDM8} as a multi-head two-way automata with an additional read-only tape, but the equivalent model appeared to be more useful in proofs. The equivalence between multi-head two-way automata and logspace machines was proofed by Cobham in his unpublished paper as we know from \cite{Ibarra71}.

Each regular realizability problem $\reg(F)$ has the corresponding GNA $M_F$, by corresponding we mean, that $\A \in \reg(F)$ iff $\A \in L(M_F)$. 

\begin{theorem}{\cite{VyaPPI}}  
	
	$\reg(F) \lelog L(M_F)$. 

\end{theorem}

And there is also a reduction in the other direction.

\begin{theorem}{\cite{VyaPPI}}  
	
	$L(M_F) \lelog \reg(F)$. 
	
\end{theorem}

If the filter contains only empty word, than GNA $M_{\{\eps\}}$ turns to a deterministic logspace machine.  In case when filter is the language of all words (under binary alphabet)  $M_{\Delta^*}$ turns to standard nondeterministic logspace machine. The intermediate case is the object of exploring of regular realizability problems.

\section{Main result}

Recall, that we call regular filter $F$ \emph{hard} if for an arbitrary regular language $R$ there exists DFST $T$, such that $R = T(F)$; in the other case we call regular filter $F$ \emph{easy}.

\begin{lemma} 
	If regular filter $F$ is hard, then the problem $\reg(F)$ is $\NL$-complete.
\end{lemma}
\begin{proof}
	The statement is obvious in case $F = \Delta^*$ -- the definition of GNA turns to the definition of nondeterministic logspace machine. In other case consider DFST $T$, such that $\Delta^* = T(F)$ and we obtain the reduction $\reg(\Delta^*) \lelog \reg(F)$ by lemma \ref{transductions-vs-logspace-reductions}. So $\reg(F)$ is $\NL$-hard. We use another DFST $T^\prime:\ F = T^\prime(\Delta^*) $  for reduction in the other direction.
\end{proof}

\begin{lemma}
	If regular filter $F$ is easy, then $\reg(F) \in \L$.
\end{lemma}
\begin{proof}
	Recall, that all easy regular languages are bounded and there for  $F \ledfst u_1^*u_2^*\ldots u_k^*$ by remark \ref{boundeddfst}. That's why we only prove the lemma in case of languages of form $u_1^*u_2^*\ldots u_k^*$. For GNA $M_F$ we built an equivalent logspace machine $M$. We equip $M$ with $k$ counters, and program it to recursively try advices $u_1^{i_1}u_2^{i_2}\ldots u_k^{i_k}$. For each tuple of indices there is only finite number of different configurations of $M_F$ on input $w$ so since the number of configurations is polynomial of $|w|$, then each index-counter can also be bounded by polynomial of $|w|$, then machine $M$ can try all possible reasonable advices and verify whether GNA $M_F$ accepts word $w$.
	
\end{proof}

\begin{theorem}
	If regular language $F$ is hard, then the problem $\reg(F)$ is $\NL$-complete; if regular language $F$ is easy, then the problem $\reg(F)$ belongs to class $\L$.
\end{theorem}

Now we compare $\reg$ and $\reg^n$ problems.

\begin{proposition}
	For each non-empty regular filter $F$ nondeterministic regular realizability problem $\reg^n(F)$ is $\NL$-complete.
\end{proposition}
\begin{proof}
	First we prove that $\reg^n(F) \in \NL$. Let NFA $\A$ be an input of the problem and $\B$ be an NFA recognizing $F$.  The NL-algorithm nondeterministically guesses paths from initial to final states in automata $\A$ and $\B$ and verifies that both paths correspond to the same word. 
	Now we prove that each nonempty language $L$ is $\NL$-hard. Let $w \in L$. We reduce the path problem to $\reg^n(L)$. Let $G(s,t)$ be an input of path problem. We built NFA $\A$ from graph $G$  by writing $\eps$ on each edge, converting $s$ to initial state and adding path labeled by $w$ from vertex $t$ to the only accepting state. So $w \in L(\A)$ iff there is a path from $s$ to $t$ in graph $G$. 
\end{proof}

\begin{prop} 
	If regular filter $F$ is easy then $\reg(F) \not\sim_{\log{}} \reg^n(F)$ modulo $\L \neq \NL$. 
\end{prop}

\end{document}